\documentclass[11pt]{article}

\usepackage[a4paper,left=20mm,top=20mm,right=20mm,bottom=25mm]{geometry}

\usepackage{amsmath, amsfonts, amssymb, amsthm}
\usepackage[noend]{algorithmic}
\usepackage[linesnumbered,ruled,vlined]{algorithm2e}

\usepackage{mathtools}
\usepackage[mathscr]{euscript}
\usepackage{wasysym}

\usepackage{graphicx}
\usepackage{xcolor}
\usepackage{enumitem}
\usepackage{authblk}

\usepackage[numbers,sort&compress]{natbib}
\let\cite\citep

\usepackage{color}

\usepackage[colorlinks]{hyperref}
\hypersetup{
  citecolor=blue,
  linkcolor=blue,
}

\usepackage[font=footnotesize,width=.85\textwidth,labelfont=bf]{caption}

\newtheorem{Theorem}{Theorem} 
\newtheorem{Lemma}{Lemma}
\newtheorem{Corollary}{Corollary}
\newtheorem{Definition}{Definition}
\newtheorem{Proposition}{Proposition}

\newtheorem{Problem}{Problem}

\DeclareMathOperator{\lca}{lca}

\newcommand{\PROBLEM}[1]{\textsc{#1}}

\newcommand{\G}{G}

\newcommand{\child}{\mathsf{child}}
\newcommand{\parent}{\mathsf{par}}

\newcommand{\AX}[1]{\textnormal{#1}}
\newcommand{\hourglass}{\mathrel{\text{\ooalign{$\searrow$\cr$\nearrow$}}}}

\providecommand{\keywords}[1]{\textbf{\textit{Keywords: }} #1}

\title{Arc-Completion of 2-Colored Best Match Graphs to
  Binary-Explainable Best Match Graphs}

\author[1,2]{David Schaller}
\author[3]{Manuela Gei{\ss}}
\author[4]{Marc Hellmuth}
\author[2,1,5-8]{Peter F.\ Stadler}

\affil[1]{Max Planck Institute for Mathematics in the Sciences,
          Inselstra{\ss}e 22, D-04103 Leipzig, Germany
          \authorcr \texttt{sdavid@bioinf.uni-leipzig.de}}

\affil[2]{Bioinformatics Group, Department of Computer Science \&
          Interdisciplinary Center for Bioinformatics, Universit{\"a}t Leipzig,
          H{\"a}rtelstra{\ss}e~16--18, D-04107 Leipzig, Germany
          \authorcr \texttt{studla@bioinf.uni-leipzig.de}}

\affil[3]{Software Competence Center Hagenberg GmbH, Hagenberg, Austria
          \authorcr \texttt{manuela.geiss@scch.at}}

\affil[4]{Department of Mathematics, Faculty of Science,
          Stockholm University, SE-10691 Stockholm, Sweden
          \authorcr \texttt{marc.hellmuth@math.su.se}}

\affil[5]{German Centre for Integrative Biodiversity
  Research (iDiv) Halle-Jena-Leipzig; Competence Center for Scalable Data
  Services and Solutions; Leipzig Research Center for Civilization
  Diseases; and Leipzig Research Center for Civilization Diseases (LIFE),
  Leipzig University, D-04103 Leipzig, Germany}

\affil[6]{Institute for Theoretical Chemistry, University of Vienna,
          W{\"a}hringerstrasse 17, A-1090 Wien, Austria}

\affil[7]{Facultad de Ciencias, Universidad National de Colombia, Sede
          Bogot{\'a}, Colombia}

\affil[8]{Santa Fe Insitute, 1399 Hyde Park Rd., Santa Fe NM 87501, USA}

\date{\ }

\begin{document}

\maketitle 

\abstract{  
  Best match graphs (BMGs) are vertex-colored digraphs that
  naturally arise in mathematical phylogenetics to formalize the notion of
  evolutionary closest genes w.r.t.\ an \textit{a priori} unknown
  phylogenetic tree.  BMGs are explained by unique least resolved trees. We
  prove that the property of a rooted, leaf-colored tree to be least
  resolved for \emph{some} BMG is preserved by the contraction of inner
  edges. For the special case of two-colored BMGs, this leads to a
  characterization of the least resolved trees (LRTs) of binary-explainable
  trees and a simple, polynomial-time algorithm for the minimum cardinality
  completion of the arc set of a BMG to reach a BMG that can be explained
  by a binary tree.
}

\bigskip
\noindent

\sloppy

\section{Introduction}

Best match graphs (BMGs) are vertex-colored digraphs that appear in
mathematical phylogenetics as a repesentation of a gene's evolutionary
closest relatives in another species \cite{Geiss:19a}. That is, given a
rooted tree $T$, a vertex (gene) $x$ in the BMG $\G(T,\sigma)$ is colored
by the species $\sigma(x)$ in which it resides, and there is an arc $(x,y)$
if there is no other gene $y'$ in species
$\sigma(y')=\sigma(y)\ne\sigma(x)$ with a later last common ancestor than
the last common ancestor $\lca_T(x,y)$ of $x$ and $y$ in $T$.  Although
rooted trees are crucial for the definition of BMGs, they are, however,
unknown in practice and we are often only left with estimates of their
BMGs.  In general, there are multiple trees that ``explain'' the same BMG.
There is, however, a unique least resolved tree (LRT) for each BMG, which
can be obtained from $T$ by contracting certain edges \cite{Geiss:19a}. The
LRTs will play a central role in this contribution. The subgraph of a BMG
induced by the vertices of some subset of colors is again a BMG. Every BMG
therefore can be viewed as the disjoint union of (the arc sets of)
2-colored BMGs.  These 2-BMGs
\cite{Geiss:19a,Korchmaros:20a,Korchmaros:20b} are bipartite and form a
common subclass of the \emph{sink-free} digraphs \cite{Cohn:02,Abrams:10}
and the \emph{bi-transitive} digraphs \cite{Das:20}.

Estimates of graphs from real-life data tend to be affected by noise and
thus typically will violate the defining properties of the desired graph
class. The solution of a corresponding graph modification problem
\cite{Natanzon:2001} therefore can by employed as a means of noise
reduction, see e.g.\ \cite{Hellmuth:15a}. The arc modification problems
(deletion, completion, and editing) for BMGs are NP-complete in general
\cite{Schaller:21b}, and remain hard even for the special case of 2 colors.

Phylogenetic trees are often considered to be binary in theory. Most
polytomies are therefore considered a limitation of the available data or
method of tree reconstruction \cite{Maddison:89,DeSalle:94} rather than a
biological reality \cite{Hoelzer:94,Slowinski:01}. In the setting of BMGs,
this distinction is important because not all BMGs can be derived from
binary gene trees. Instead, \emph{binary-explainable} BMGs (beBMGs) form a
proper subclass \cite{Schaller:21c} that is distinguished by a single
forbidden induced subgraph, the \emph{hourglass}, from other BMGs
\cite{Schaller:21a}. The arc modification problems for beBMGs are
NP-complete \cite{Schaller:21b,Schaller:21c} as well.

In the context of correcting empirical best match data, it is natural to
ask whether the problem of modifying a BMG to a beBMG is as difficult as
the general case. It is, in fact, not unusual that graph modification
problems that are hard in general become easy when the input is confined to
a -- usually restrictive -- class of graphs, see e.g.\
\cite{Liu:11,Gao:13}. Here we show that the problem of completing a
2-colored BMG to a beBMG can indeed be solved in polynomial time.

To prove this result we make use of the fact that every BMG is associated
with a unique \emph{least resolved tree} (LRT).
Thm.~\ref{thm:edge-contr-supergraph} shows that the property of being the
LRT for some BMG is preserved under contraction of inner edges. This
observation leads to the explicit construction of a ``collapsed tree''
from the LRT of the input BMG $(G,\sigma)$ which not only is the LRT of a
2-colored beBMG but also minimizes the number of arcs that need to be
inserted to obtain a beBMG from $(G,\sigma)$. The construction does not
generalize to more than $2$ colors.

\section{Notation}

We consider simple directed graphs (digraphs) $\G=(V,E)$ with vertex set
$V$ and \emph{arc} set
$E\subseteq V\times V \setminus \{(v,v)\mid v\in V\}$ and rooted
(undirected) trees $T$ with root $\rho$.  Correspondingly, we write $(x,y)$
for directed arcs from $x$ to $y$, and $xy$ for undirected tree edges.
Given a tree $T$, we write $V(T)$ and $E(T)$ for its set of vertices and
edges, resp., $L(T)$ for the set of leaves, and $V^0(T)=V(T)\setminus L(T)$
for the set of inner vertices.

A vertex coloring of a graph is a map $\sigma:V\to S$, where $S$ is
a non-empty set of colors.  A vertex coloring of $\G$ is \emph{proper} if
$\sigma(x)\ne\sigma(y)$ for all $(x,y)\in E(\G)$.
We will also consider \emph{leaf-colorings} $\sigma\colon L(T)\to S$ for 
trees $T$, and denote by $(\G,\sigma)$ and $(T,\sigma)$ vertex-colored graphs 
and leaf-colored trees, respectively.

Given a rooted tree, we write $x \preceq_T y$ if $y$ is an \emph{ancestor}
of $x$, i.e., if $y$ lies along the unique path from $\rho$ to $x$ in $T$.
We write $x \prec_T y$ if $x \preceq_T y$ and $x\ne y$. The relation
$\preceq_T$ is a partial order on $T$. If $xy\in E(T)$ and $x\prec_T y$,
then $y$ is the unique \emph{parent} of $x$, denoted by $\parent_T(x)$, and
$x$ a \emph{child} of $y$. The set of children of a vertex $u\in V(T)$ is
denoted by $\child_T(u)$.  A rooted tree $T$ is phylogenetic if every inner
vertex $x\in V^0(T)$ has at least two children. All trees in this
contribution are assumed to be phylogenetic.  Furthermore, we write $T(u)$
for the subtree rooted in $u$, i.e.,
$V(T(u))=\{y\in V(T) \mid y\preceq_T u\}$. The \emph{last common ancestor}
of a non-empty subset $A\subseteq V(T)$ is the unique $\preceq_T$-minimal
vertex of $T$ that is an ancestor of every $u\in A$. For convenience, we
write $\lca(x,y,\dots)$ instead of $\lca(\{x,y,\dots\})$.

A triple $xy|z$ is a rooted tree with the three leaves $x$, $y$, and $z$
such that $\lca(x,y)\prec \lca(x,y,z)$. If $e\in E(T)$, we denote by $T_e$
the tree obtained by contracting the edge $e$. We will only be
interested in contractions of inner edges, i.e., those that preserve the
leaf set. We say that $T$ \emph{displays} a tree $T'$, in symbols
$T'\le T$, if $T'$ can be obtained from $T$ as the minimal subtree of $T$
that connects all elements in $L(T')$ with root $\lca_{T}(L(T'))$ and by
suppressing all inner vertices that only have one child left.

\section{Best Match Graphs, Least Resolved Trees, and Binary-Explainable
  BMGs}

In this section, we first summarize some properties of best match graphs and
their least resolved trees. We then show that the contraction of inner
edges in least resolved trees always leads to least resolved
trees. Furthermore, we recall some properties of binary-explainable best
match graphs that will be needed later.

\begin{Definition}
  Let $(T,\sigma)$ be a leaf-colored tree. A leaf $y\in L(T)$ is a
  \emph{best match} of the leaf $x\in L(T)$ if $\sigma(x)\neq\sigma(y)$ and
  $\lca(x,y)\preceq_T \lca(x,y')$ holds for all leaves $y'$ of color
  $\sigma(y')=\sigma(y)$.
  \label{def:BMG}
\end{Definition}
Given $(T,\sigma)$, the graph $\G(T,\sigma) = (V,E)$ with vertex set
$V=L(T)$, vertex-coloring $\sigma$, and with arcs $(x,y)\in E$ if and only
if $y$ is a best match of $x$ w.r.t.\ $(T,\sigma)$ is called the \emph{best
  match graph} (BMG) of $(T,\sigma)$ \cite{Geiss:19a}:

\begin{Definition}\label{def:BestMatchGraph}
  An arbitrary vertex-colored graph $(\G,\sigma)$ is a \emph{best match
    graph (BMG)} if there exists a leaf-colored tree $(T,\sigma)$ such that
  $(\G,\sigma) = \G(T,\sigma)$. In this case, we say that $(T,\sigma)$
  \emph{explains} $(\G,\sigma)$.
\end{Definition}

\begin{Proposition} \cite[Lemma~8]{Schaller:21a}
  If $T_A$ is obtained from a tree $T$ by contracting all edges in a subset
  $A$ of inner edges in $T$, then $\G(T,\sigma)\subseteq \G(T_A,\sigma)$.
  \label{prop:TA}
\end{Proposition}

An edge $e$ of a leaf-colored tree is \emph{redundant} (w.r.t.\
$(\G,\sigma)$) if it can be contracted without affecting the BMG, i.e., if
$\G(T,\sigma)=\G(T_e,\sigma)$.

\begin{Definition}
  A leaf-colored tree $(T,\sigma)$ is \emph{least resolved} if there is no
  non-empty subset $A\subseteq E(T)$ such that
  $\G(T,\sigma)=\G(T_A,\sigma)$.
\end{Definition}
We define the notion of being least resolved here as a property of the tree
$(T,\sigma)$ alone. Of course, every least resolved tree is also
\emph{least resolved w.r.t.\ some BMG}, namely the (uniquely defined)
graph $\G(T,\sigma)$.

It is shown in \cite{Geiss:19a}
that $(T,\sigma)$ is least resolved if and only if it does not contain a
redundant edge. In particular, we have 
\begin{Proposition} \cite[Thm.~8]{Geiss:19a}
  Every BMG $(\G,\sigma )$ is explained by a unique least resolved tree
  (LRT), which is obtained from an arbitrary tree $(T,\sigma)$ that explains 
  $(\G,\sigma)$ by contraction of all redundant edges of $(T,\sigma)$.
  \label{prop:LRTuniq}
\end{Proposition}
In particular, therefore, there is a bijection between BMGs and LRTs.
Surprisingly, the property of being least resolved for some BMG is
preserved under contraction of inner edges of $T$. 

\begin{Theorem}
  \label{thm:edge-contr-supergraph}
  Suppose $(T,\sigma)$ is least resolved 
  and let $A$ be a set of inner
  edges of $T$, and denote by $T_A$ the tree obtained from a tree $T$ by
  contracting all edges in $A$. Then $(T_A,\sigma)$ is again least
  resolved. 
\end{Theorem}
\begin{proof}
  Assume that $(T,\sigma)$ is least resolved, i.e., it does not contain any
  redundant edges, and set $(\G,\sigma)\coloneqq\G(T,\sigma)$.  Lemma~7 in
  \cite{Schaller:21a} states that an inner edge $e=uv$ with $v\prec_{T} u$
  in $(T,\sigma)$ is non-redundant if and only if there is an arc
  $(a,b)\in E(\G)$ such that $\lca_{T}(a,b)=v$ and
  $\sigma(b)\in \sigma( L(T(u)) \setminus L(T(v)) )$.  The statement
  trivially holds if $(T,\sigma)$ has at most one inner edge. Hence, we
  assume that $(T,\sigma)$ has at least two distinct inner edges $e=uv$
  and $e'$. We show that every non-redundant edge $e$ in $T$ remains 
  non-redundant in $T_{e'}$. Thus, let $e$ be a non-redundant edge in $T$. 
  Hence, there is an arc $(a,b)\in E(\G)$ such that $\lca_{T}(a,b)=v$ and
  $\sigma(b)\in \sigma( L(T(u)) \setminus L(T(v)) )$. Now consider the tree
  $T_{e'}$ obtained from $T$ by contraction of the inner edge $e'\ne e$.
  Clearly, we also have $\lca_{T_{e'}}(a,b)=v$ and
  $\sigma(b)\in \sigma( L(T_{e'}(u)) \setminus L(T_{e'}(v))
  )$. Prop.~\ref{prop:TA} implies $\G(T,\sigma)\subseteq\G(T_{e'},\sigma)$,
  and thus, $(a,b)\in E(\G(T_{e'},\sigma))$.  Making again use of the
  characterization of redundant edges in \cite[Lemma~7]{Schaller:21a}, we
  conclude that $e$ is non-redundant in $(T_{e'},\sigma)$. 
  
  Since both $e$ and $e'$ were chosen arbitrarily, we observe that the
  contraction of a single inner edge does not produce new redundant
  edges. We can therefore apply this argument for each step in the consecutive
  contraction of all edges in $A$ (in an arbitrary order) to conclude that
  $(T_A,\sigma)$ does not contain redundant edges. Therefore, 
  Prop.~\ref{prop:LRTuniq} implies that $(T_A,\sigma)$
  is least resolved. 
\end{proof}

\begin{Corollary}
  If $(T,\sigma)$ is least resolved and $A$ is a non-empty set
  of inner edges of $T$, then $\G(T,\sigma)\subsetneq\G(T_A,\sigma)$.
\end{Corollary}
\begin{proof}
  By Prop.~\ref{prop:TA}, we have $\G(T,\sigma)\subseteq\G(T_A,\sigma)$. By 
  Thm.~\ref{thm:edge-contr-supergraph}, $(T_A,\sigma)$ is least resolved.
  Since the LRT of a BMG is unique (cf.\ Prop.~\ref{prop:LRTuniq}), we have
  $\G(T,\sigma)\neq\G(T_A,\sigma)$.
\end{proof}
As another immediate consequence of Thm.~\ref{thm:edge-contr-supergraph}
and uniqueness of the LRT of a BMG (Prop.~\ref{prop:LRTuniq}), we obtain
\begin{Corollary}
  If $e$ and $e'$ are two distinct inner edges of a least resolved tree
  $(T,\sigma)$, then $\G(T_{e},\sigma)\ne\G(T_{e'},\sigma)$.
\end{Corollary}

Let us now turn to the subclass of BMGs that can be explained by a
binary tree. 
\begin{Definition}\label{def:beBMG}
  A \emph{binary-explainable BMG} (\emph{beBMG}) is a BMG $(\G,\sigma)$
  such that there is a binary leaf-colored tree $(T,\sigma)$ that explains
  $(\G,\sigma)$.
\end{Definition}
As shown in \cite{Schaller:21a}, beBMGs can be characterized among BMGs
by means of a simple forbidden colored induced subgraph:
\begin{Definition}
  An \emph{hourglass} in a properly vertex-colored graph $(\G,\sigma)$,
  denoted by $[xy \hourglass x'y']$, is a subgraph $(\G[Q],\sigma_{|Q})$
  induced by a set of four pairwise distinct vertices
  $Q=\{x, x', y, y'\}\subseteq V(\G)$ such that (i)
  $\sigma(x)=\sigma(x')\ne\sigma(y)=\sigma(y')$, (ii) $(x,y),(y,x)$ and 
  $(x'y'),(y',x')$ are bidirectional arcs
  in $\G$, (iii) $(x,y'),(y,x')\in E(\G)$, and (iv)
  $(y',x),(x',y)\notin E(\G)$.
\end{Definition}

An hourglass together with a (non-binary) tree explaining it is 
illustrated in Fig.~\ref{fig:forb-sg-and-hourglass.pdf}(A).
A properly vertex-colored digraph that does not contain an hourglass as an 
induced subgraph is called \emph{hourglass-free.}

\begin{figure}[htb]
  \centering
  \includegraphics[width=0.75\textwidth]{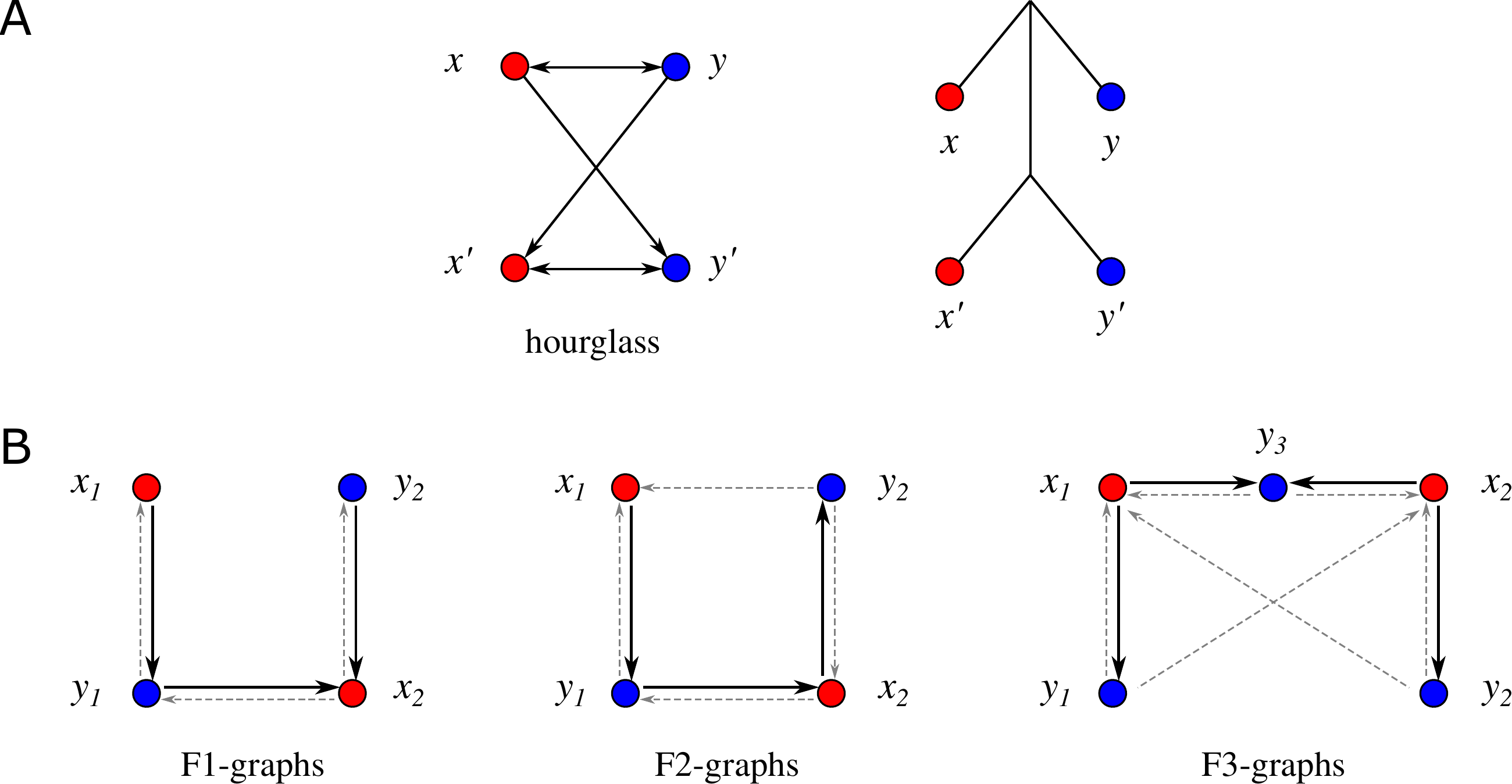}
  \caption{(A): An hourglass as the characteristic forbidden induced 
    subgraph of 
    beBMGs and its non-binary explaining tree.
    (B): The three classes of forbidden induced subgraphs of 2-colored BMGs 
    (see 
    Def.~\ref{def:forbidden-subgraphs} below). The 
    gray dashed arcs may or may not exist.}
  \label{fig:forb-sg-and-hourglass.pdf}
\end{figure}

\begin{Proposition}{\cite[Lemma~31 and Prop.~8]{Schaller:21a}}
  \label{prop:binary-iff-subtree-colors}
  For every BMG $(\G,\sigma)$ explained by a tree $(T,\sigma)$, the following 
  three statements are equivalent:
  \begin{enumerate}
    \item $(\G,\sigma)$ is binary-explainable.
    \item $(\G,\sigma)$ is hourglass-free.
    \item There is
    no vertex $u\in V^0(T)$ with three distinct children $v_1$, $v_2$,
    and $v_3$ and two distinct colors $r$ and $s$ satisfying
    \begin{enumerate}
      \item $r\in\sigma(L(T(v_1)))$, 
      $r,s\in\sigma(L(T(v_2)))$, 
      and $s\in\sigma(L(T(v_3)))$, and
      \item $s\notin\sigma(L(T(v_1)))$, and 
      $r\notin\sigma(L(T(v_3)))$.
    \end{enumerate}
  \end{enumerate}
\end{Proposition}

It is worth noting that the LRTs of beBMGs are usually not binary. In fact,
it is shown in \cite{Schaller:21c} that, for a beBMG $(G,\sigma)$, there
exists a unique \emph{binary refinable tree} (BTR) $B(\G,\sigma)$ with the
property that every binary tree $(T,\sigma)$ that displays $B(\G,\sigma)$
explains $(\G,\sigma)$. The BRT is in general much better resolved than the
LRT of $(G,\sigma)$.

\section{Two-Colored BMGs} 

Let us now briefly focus on 2-colored BMGs (2-BMGs).  Since arcs in BMG can
only connect vertices with different colors, every 2-BMG is
bipartite. Furthermore, every leaf $x$ in a tree with two leaf colors has
at least one best match $y$. Every 2-BMG is therefore \emph{sink-free},
i.e., every vertex has at least one out-neighbor.  Furthermore,
\citet{Schaller:21b} showed that the following graphs (see also
Fig.~\ref{fig:forb-sg-and-hourglass.pdf}(B)) are forbidden induced
subgraphs for 2-BMGs.

\begin{Definition}[F1-, F2-, and F3-graphs]\par\noindent
  \begin{itemize}
    \item[\AX{(F1)}] A properly 2-colored graph on four distinct vertices
    $V=\{x_1,x_2,y_1,y_2\}$ with coloring
    $\sigma(x_1)=\sigma(x_2)\ne\sigma(y_1)=\sigma(y_2)$ is an
    \emph{F1-graph} if $(x_1,y_1),(y_2,x_2),(y_1,x_2)\in
    E$ and $(x_1,y_2),(y_2,x_1)\notin E$.
    \item[\AX{(F2)}] A properly 2-colored graph on four distinct vertices
    $V=\{x_1,x_2,y_1,y_2\}$ with coloring
    $\sigma(x_1)=\sigma(x_2)\ne\sigma(y_1)=\sigma(y_2)$ is an
    \emph{F2-graph} if $(x_1,y_1),(y_1,x_2),(x_2,y_2)\in E$ and
    $(x_1,y_2)\notin E$.
    \item[\AX{(F3)}] A properly 2-colored graph on five distinct vertices
    $V=\{x_1,x_2,y_1,y_2,y_3\}$ with coloring
    $\sigma(x_1)=\sigma(x_2)\ne\sigma(y_1)=\sigma(y_2)=\sigma(y_3)$ is an
    \emph{F3-graph} if \newline
    $(x_1,y_1),(x_2,y_2),(x_1,y_3),(x_2,y_3)\in E$ and
    $(x_1,y_2),(x_2,y_1)\notin E$.
  \end{itemize}
  \label{def:forbidden-subgraphs}
\end{Definition}

\begin{Proposition} \cite[Thm.~4.4]{Schaller:21b}
  \label{prop:forb-subgraph-charac}
  A properly 2-colored graph is a BMG if and only if it is sink-free and
  does not contain an induced F1-, F2-, or F3-graph.
\end{Proposition}

A peculiar property of 2-BMGs is that their LRTs can be constructed
efficiently by recursively decomposing an input 2-BMG into non-trivial
induced subgraphs and individual vertices \cite{Schaller:21c}. Although we
will not need this construction here, one of its corner stones plays an
important role below.
\begin{Definition}[Support Leaves]
  For a given tree $T$, the set $S_{u} \coloneqq \child_T(u)\cap L(T)$ is
  the set of all \emph{support leaves} of vertex $u\in V(T)$.
\end{Definition}
We note in passing that every inner vertex $u$ of the LRT of a 2-BMG
$(\G,\sigma)$, with the possible exception of the root $\rho$, has a
non-empty set of support leaves $S_{u}$, and $S_{\rho}\ne\emptyset$ if and
only if $(\G,\sigma)$ is connected \cite{Schaller:21x}. In the following, we
will make use of a connection between a 2-BMG and its LRT:

\begin{Lemma}
  \label{lem:2-BMG-outneighbors-subtree}
  Let $(\G,\sigma)$ be a 2-BMG, $(T,\sigma)$ its LRT and
  $x,y\in L(T)=V(\G)$.  Then $(x,y)\in E(\G)$ if and only if
  $\sigma(x)\ne \sigma(y)$ and $y\in L(T(\parent_T(x)))$.
\end{Lemma}
\begin{proof}
  First note that, since $(\G,\sigma)$ is 2-colored, $(T,\sigma)$ has at
  least two leaves and $u\coloneqq\parent_T(x)$ is always defined.  
  First, assume $\sigma(x)\ne\sigma(y)$,
  and thus $x\ne y$, and let $y\in L(T(u))$.
  Since $x$ is a child of $u$, we have $\lca_{T}(x,y)=u$.
  Moreover, since $u$ is the parent of $x$, there is no vertex $y'$ of color 
  $\sigma(y)$ such that
  $\lca_{T}(x,y')\prec_T\lca_T(x,y)=u$. Hence, $y$ is a best match of $x$,
  i.e., $(x,y)\in E(\G)$. 
  
  Now suppose, for contraposition, that
  $\sigma(x)=\sigma(y)$ or $y\notin L(T(u))$. 
  If $\sigma(x)=\sigma(y)$, then, by definition, $(x,y)\notin E(G)$.
  If $y\notin L(T(u))$, then
  $u\prec_{T}\rho_T$.  Hence, we can apply Cor.~1 in \cite{Schaller:21x} to
  the inner vertex $u$ to conclude that $|\sigma(L(T(u)))|>1$, i.e., the
  subtree $L(T(u))$ contains both colors.  Thus, we can find a vertex $y'$
  of color $\sigma(y)$ such that
  $\lca_{T}(x,y')\preceq_{T}u\prec_{T}\lca_{T}(x,y)$ which implies that
  $(x,y)\notin E(\G)$.
\end{proof}
As an immediate consequence, we find
\begin{Corollary}
  \label{cor:reciprocal-iff-same-parent}
  Let $(\G,\sigma)$ be a 2-BMG, $(T,\sigma)$ its LRT and
  $x,y\in V(\G)=L(T)$.  Then $(x,y),(y,x)\in E(\G)$ if and only if
  $\sigma(x)\ne \sigma(y)$ and $\parent_T(x)=\parent_T(y)$.
\end{Corollary}

\section{Completion of a 2-BMG to a 2-beBMG}

Writing $\G+F\coloneqq (\G, E\cup F)$ for a graph $\G=(V,E)$ and arc set
$F\subseteq V\times V \setminus \{(v,v)\mid v\in V\}$, consider the
following graph completion problem:

\begin{Problem}[\PROBLEM{2-BMG Completion restricted to 
    Binary-Explainable Graphs (2-BMG CBEG)}]\ \\
  \begin{tabular}{ll}
    \emph{Input:}    & A properly 2-colored digraph $(\G =(V,E),\sigma)$
    and an integer $k$.\\
    \emph{Question:} & Is there a subset $F\subseteq V\times V \setminus
    (\{(v,v)\mid v\in V\} \cup E)$ such that\\
    &  $|F|\leq k$ and $(\G+ F,\sigma)$
    is a binary-explainable 2-BMG?
  \end{tabular}
\end{Problem}

In the general case, \PROBLEM{2-BMG CBEG} is NP-complete
\cite[Cor.~5.11]{Schaller:21c}. Here we are interested in the restriction
of the \PROBLEM{2-BMG CBEG} problem with BMGs as input.

The following result holds for BMGs and their completions to beBMGs with an
arbitrary number of colors. 
\begin{Lemma}
  \label{lem:fill-houglass}
  Let $(\G',\sigma)$ be a completion of a BMG $(\G,\sigma)$ to a beBMG, and
  let $[xy\hourglass x'y']$ be an induced hourglass in $(\G,\sigma)$.  Then
  $(\G',\sigma)$ contains both arcs $(x',y)$ and $(y',x)$.
\end{Lemma}
\begin{proof}
  It is shown in \cite[Obs.~1]{Geiss:19a} that the subgraphs of a BMG induced
  by all vertices with any two given colors is a 2-BMG. Since
  $(\G',\sigma)$ is a (binary-explainable) BMG, all of its 2-colored
  induced subgraphs are therefore 2-BMGs.  By assumption, $(\G,\sigma)$ is
  not binary-explainable since it contains the hourglass
  $[xy\hourglass x'y']$ as an induced subgraph (cf.\
  Prop.~\ref{prop:binary-iff-subtree-colors}). The hourglass contains all
  possible arcs between vertices of different colors except $(x',y)$ and
  $(y',x)$. Since $(\G',\sigma)$ contains no hourglass, and $\G'$ is a
  completion of $G$, i.e., $E(\G)\subseteq E(\G')$, we conclude that
  $(\G',\sigma)$ contains at least one of the arcs $(x',y)$ and $(y',x)$.
  
  Assume for contradiction that, w.l.o.g., $(\G',\sigma)$ only contains
  $(x',y)$.  We have $(y',x'),(y,x)\in E(\G')$ and
  $\sigma(y')=\sigma(y)\ne\sigma(x')=\sigma(x)$ by the definition of
  hourglasses, and by assumption $(x',y)\in E(\G')$ and
  $(y',x)\notin E(\G')$.  Hence, the four vertices $x,x',y,y'$ induce an
  F2-graph in $(\G',\sigma)$.  By Prop.~\ref{prop:forb-subgraph-charac},
  the 2-colored subgraph of $(\G',\sigma)$ induced by the two colors
  $\sigma(x)$ and $\sigma(y)$ is not a BMG. Consequently, $(\G',\sigma)$ is
  not a BMG either; a contradiction.  Hence, $(\G',\sigma)$ contains both
  arcs $(x',y)$ and $(y',x)$.
\end{proof}

\begin{Definition}
  Let $(T,\sigma)$ be a tree with a 2-colored leaf set, i.e.,
  $|\sigma(L(T))|=2$. Denote by $(T^*,\sigma)$ the \emph{collapsed tree}
  obtained from $(T,\sigma)$ by contraction of all inner edges in $T(u)$
  for all $u\in V^0(T)$ that have support leaves of both colors.
\end{Definition}

In other words, $(T^*,\sigma)$ is obtained from $(T,\sigma)$ by
collapsing every subtree $T(u)$ to a star if $u$ has support leaves
of both colors.

\begin{Lemma}
  The collapsed tree $(T^*,\sigma)$ of $(T,\sigma)$ is uniquely defined and
  can be computed from $(T,\sigma)$ in $O(|V(T)|)$-time.
  \label{lem:collaps}
\end{Lemma}
\begin{proof}
  The collapsed tree $(T^*,\sigma)$ is well-defined because whenever
  $v\prec_T u$, then collapsing the subtree $T(v)$ to a star does not
  change the set of support leaves $S_u$. Similarly, collapsing $T(v)$ if
  $v$ is not $\prec_T$-comparable with $u$ does not change $S_u$. Thus
  $(T^*,\sigma)$ is uniquely defined. To see that $(T^*,\sigma)$ can be
  computed in $O(|V(T)|)$ operations, we observe that it suffices to
  collapse all subtrees $T(u)$ such that $u\in V^0(T)$ has support leaves
  of both colors and there is no $u'\prec_{T}u$ with this property, i.e.,
  $u$ is $\preceq_{T}$-maximal in that sense. These vertices $u$ for which
  $T(u)$ is replaced by a star are found by a top-down traversal of $T$ and
  evaluating $|\sigma(S_u)|$, all of which can be computed in linear total
  time.
\end{proof}

As an immediate consequence of the uniqueness of $T^*$ and the construction
in the second part of the proof of Lemma~\ref{lem:collaps}, we obtain
\begin{Corollary}
  The collapsed tree $(T^{**},\sigma)$ of a collapsed tree $(T^{*},\sigma)$
  satisfies $T^{**}=T^{*}$.
\end{Corollary}

\begin{Lemma}
  If $(T^*,\sigma)$ is the collapsed tree of an LRT $(T,\sigma)$ with
  2-colored leaf set, then $\G(T^*,\sigma)$ is binary-explainable.
  \label{lem:T*be}
\end{Lemma}
\begin{proof}
  Since the collapsed tree $(T^*,\sigma)$ is obtained from the LRT
  $(T,\sigma)$ by contraction of edges,
  Thm.~\ref{thm:edge-contr-supergraph} implies that $(T^*,\sigma)$ is also
  least resolved.  Now suppose, for contradiction, that $\G(T^*,\sigma)$ 
  is not binary-explainable. By,   
  Prop.~\ref{prop:binary-iff-subtree-colors}(3),
  $(T^*,\sigma)$ has
  a vertex $u\in V^0(T^*)$ with three distinct children $v_1$, $v_2$, and
  $v_3$ and two distinct colors $r$ and $s$ satisfying (i)
  $r\in\sigma(L(T^*(v_1)))$, $r,s\in\sigma(L(T^*(v_2)))$, and
  $s\in\sigma(L(T^*(v_3)))$, and (ii) $s\notin\sigma(L(T^*(v_1)))$, and
  $r\notin\sigma(L(T^*(v_3)))$.  Since $(\G,\sigma)$ is only 2-colored, the
  latter arguments imply that
  $|\sigma(L(T^*(v_1)))|=|\sigma(L(T^*(v_3)))|=1$ and
  $|\sigma(L(T^*(v_2))|=2$.  Since moreover $(T^*,\sigma)$ is least
  resolved and none of the vertices $v_1$, $v_2$, and $v_3$ is the root of
  $T^*$, we can apply Cor.~1 in \cite{Schaller:21x} to conclude that $v_1$
  and $v_2$ are leaves, and that $v_3$ is an inner vertex, respectively. In
  particular, $\sigma(v_1)=r\ne s=\sigma(v_3)$.  Hence, $T^*(u)$ is not a
  star tree and $u$ has support leaves of both colors in $T^*$; a
  contradiction to its construction.  Therefore, we can apply
  Prop.~\ref{prop:binary-iff-subtree-colors} to conclude that
  $\G(T^*,\sigma)$ is binary-explainable.
\end{proof}

\begin{Theorem}
  \label{thm:2BMG-CBEG-BMG-input}
  The optimization version of \PROBLEM{2-BMG CBEG} with a 2-BMG
  $(\G,\sigma)$ as input has the unique solution
  $F\coloneqq E(\G(T^*,\sigma))\setminus E(\G)$, where $(T^*,\sigma)$ is
  the collapsed tree of the LRT $(T,\sigma)$ of $(\G,\sigma)$.
\end{Theorem}
\begin{proof}
  First note that the optimization version of \PROBLEM{2-BMG CBEG} always
  has a solution. To see this, consider the complete bipartite and properly 
  2-colored graph $(\G',\sigma)$
  with vertex set $V(\G)$. This graph is explained by the star tree with
  leaf set $V(\G)$.  Moreover, $(\G',\sigma)$ is clearly hourglass-free
  since hourglasses require non-arcs (between vertices of distinct colors). By
  Prop.~\ref{prop:binary-iff-subtree-colors}, the BMG $(\G',\sigma)$ is
  binary-explainable.
  
  Now consider the collapsed tree $(T^*,\sigma)$ of $(T,\sigma)$.  Since
  $T^*$ is obtained from $T$ by contraction of inner edges,
  Prop.~\ref{prop:LRTuniq} implies
  $(\G,\sigma)=\G(T,\sigma)\subseteq \G(T^*,\sigma)\eqqcolon(\G^*,\sigma)$. 
  Furthermore,
  $(\G^*,\sigma)$ is binary-explainable by Lemma~\ref{lem:T*be}. Therefore,
  $(\G^*,\sigma)$ is a valid completion of $(\G,\sigma)$ to a beBMG.
  
  We continue by showing the existence of certain arcs in every (not
  necessarily optimal) completion $(\G',\sigma)$ of $(\G,\sigma)$ to a
  beBMG.  To this end, consider a $\preceq_{T}$-maximal vertex $u$ such
  that the subtree $T(u)$ is not a star tree and $u$ has support leaves
  $S_u$ of both colors in $T$.  We will make frequent use of the fact that
  $E(\G)\subseteq E(\G')$.  We consider the following cases in order to
  show that all arcs between vertices $x, y\in L(T(u))$ with
  $\sigma(x)\ne\sigma(y)$ exist in $(\G',\sigma)$:
  \begin{description}
    \item[(i)] $x,y\in S_u$,
    \item[(ii)] $x\in L(T(u))\setminus S_u$ and $y\in S_u$, and
    \item[(iii)] $x,y\in L(T(u))\setminus S_u$.
  \end{description}
  In Case~(i), the leaves $x$ and $y$ are both children of $u$. Together with 
  Cor.~\ref{cor:reciprocal-iff-same-parent}, this implies $(x,y),(y,x)\in 
  E(\G)\subseteq E(\G')$.
  
  In Case~(ii), we can find a vertex $x'\in S_u$ of color 
  $\sigma(x)$ since $S_u$ 
  contains vertices of both colors. As in Case~(i), we have 
  $(x',y),(y,x')\in E(\G)\subseteq E(\G')$.
  Since $x\in L(T(u))\setminus S_u$, we can conclude that $v\coloneqq 
  \parent_T(x)\prec_T u$ by the definition of support leaves.
  Hence, the inner vertex $v$ is not the root of $T$ and we can apply Cor.~1 in 
  \cite{Schaller:21x} to conclude that the subtree $T(v)$ of the inner vertex 
  $v$ contains both colors.
  The latter together with Lemma~10 in \cite{Geiss:19b} implies that there are 
  arcs $(x'',y''),(y'',x'')\in 
  E(\G)\subseteq E(\G')$ with $x'',y''\in L(T(v))$ and 
  $\sigma(x)=\sigma(x'')\ne \sigma(y)=\sigma(y'')$. Note that $x=x''$ is 
  possible.
  Since $x,x'',y''$ in $L(T(v))\subset L(T(u))$, $x',y\in L(T(u))\setminus 
  L(T(v))$ and $v\prec_T u$, we can apply 
  Lemma~\ref{lem:2-BMG-outneighbors-subtree} to conclude that
  $(x',y''),(y,x),(y,x'')\in E(\G)\subseteq E(\G')$ and 
  $(y'',x'),(x,y),(x'',y)\notin E(\G)\subseteq E(\G')$.
  Together with $(x',y),(y,x'),(x'',y''),(y'',x'')\in E(\G)$ and the coloring, 
  this implies that $x',y,x'',y''$ induce an hourglass $[x'y\hourglass x''y'']$ 
  in $(\G,\sigma)$.
  By Lemma~\ref{lem:fill-houglass}, we have arcs $(x'',y),(y'',x')\in E(\G')$.
  If $x=x''$, we immediately obtain $(x,y), (y,x)\in E(\G')$.
  Now suppose $x\ne x''$, i.e., it remains to show that $(x,y)\in E(\G')$.
  Thus assume, for contradiction, that $(x,y)\notin E(\G')$.
  Lemma~\ref{lem:2-BMG-outneighbors-subtree} together with 
  $\sigma(x)\ne\sigma(y'')$ and $y''\in L(T(\parent_T(x)=v))$ implies that 
  $(x,y'')\in E(\G)\subseteq E(\G')$.
  Hence, we have the arcs $(x,y''), (y'',x'), (x',y)\in E(\G')$ but 
  $(x,y)\notin 
  E(\G')$, i.e., $x,x',y,y''$ induce a forbidden F2-graph. Together with 
  Prop.~\ref{prop:forb-subgraph-charac}, this is a contradiction to 
  $(\G',\sigma)$ being a 2-BMG. Therefore, we conclude that $(x,y)\in E(\G')$.
  
  In Case~(iii), we have $x,y\in L(T(u))\setminus S_u$. We can find two 
  vertices $x',y'\in S_u$, which are distinct from $x$ and $y$ and satisfy 
  $\sigma(x)=\sigma(x')\ne\sigma(y)=\sigma(y')$.
  From Cases~(i) and~(ii), we obtain $(x',y'),(y',x')\in E(\G')$ and 
  $(x',y),(y,x'),(x,y'),(y',x)\in E(\G')$, respectively.
  Now assume for contradiction that $(x,y)\notin E(\G')$.
  Thus, we have $(x,y'),(y',x'),(x',y)\in E(\G')$ and $(x,y)\notin 
  E(\G')$, i.e., $x,x',y,y'$ induce a forbidden F2-graph in $(\G',\sigma)$;
  a contradiction to $(\G',\sigma)$ being a 2-BMG.
  Hence, we conclude that $(x,y)\in E(\G')$. The existence of the arc 
  $(y,x)\in E(\G')$ can be shown by analogous arguments.
  
  We will now show that $E(\G^*)\subseteq E(\G')$ for every (not
  necessarily optimal) completion $(\G',\sigma)$ of the 2-BMG $(\G,\sigma)$
  to a beBMG.  To this end, consider an arbitrary arc $(x,y)\in E(\G^*)$.
  If $(x,y)\in E(\G)$, then $(x,y)\in E(\G')$ follows immediately.  Now
  assume that $(x,y)\in F=E(\G^*)\setminus E(\G)$.  Since $(\G,\sigma)$ is
  a 2-BMG and thus properly-colored and sink-free (cf.\
  Prop.~\ref{prop:forb-subgraph-charac}), there must be a vertex $y'$ of
  color $\sigma(y)$ such that $(x,y')\in E(\G)$. Since $(x,y)\notin E(\G)$,
  we have $\lca_T(x,y')\prec_T \lca_T(x,y)$ and thus the LRT $(T,\sigma)$
  displays the triple $xy'|y$.  However, $(x,y), (x,y')\in E(\G^*)$ implies
  that $(T^*,\sigma)$ does not display the triple $xy'|y$, i.e., all edges
  on the path from $\lca_{T}(x,y')$ to $\lca_{T}(x,y)$ have been
  contracted.  Therefore, there is a $\preceq_{T}$-maximal inner vertex
  $u\in V^0(T)$ such that $x,y\in L(T(u))$, $T(u)$ is not a star tree and $u$ 
  has support
  leaves of both colors in $T$.  By the arguments above, we can conclude
  that $(x,y)\in E(\G')$.
  
  In summary, $F$ is a solution for \PROBLEM{2-BMG CBEG} with the 2-BMG
  $(\G,\sigma)$ (and some integer $k\ge|F|$) as input, and $F\subseteq F'$
  for every other solution $F'= E(\G')\setminus E(\G)$.  Therefore, we
  conclude that $F$ is the unique optimal solution.
\end{proof}

As a direct cosequence of Thm.~\ref{thm:2BMG-CBEG-BMG-input}, the fact that
LRTs can be constructed in $O(|V| + |E|\log^2 |V|)$ (cf.\
\cite{Schaller:21x}) and Lemma~\ref{lem:collaps}, we have
\begin{Corollary}
  \PROBLEM{2-BMG CBEG} with a 2-BMG as input can be solved in
  $O(|V| + |E|\log^2 |V|)$ time. 
\end{Corollary}
We also immediately obtain a characterization of the LRTs of 2-beBMGs.
\begin{Corollary}
  A 2-colored least resolved tree $(T,\sigma)$ is the LRT of 2-beBMG if and
  only if it is a collapsed tree. 
\end{Corollary}

\section{Concluding Remarks} 

Starting from the observation that the property of being least resolved is
preserved under contraction of inner edges, we have obtained a
characterization of the LRTs that explain 2-colored beBMGs. The
construction of these ``collapsed trees'' corresponds to the completion of
BMGs to beBMGs, resulting in a simple, polynomial-time algorithm for this
problem.

\begin{figure}[htb]
  \centering
  \includegraphics[width=0.85\textwidth]{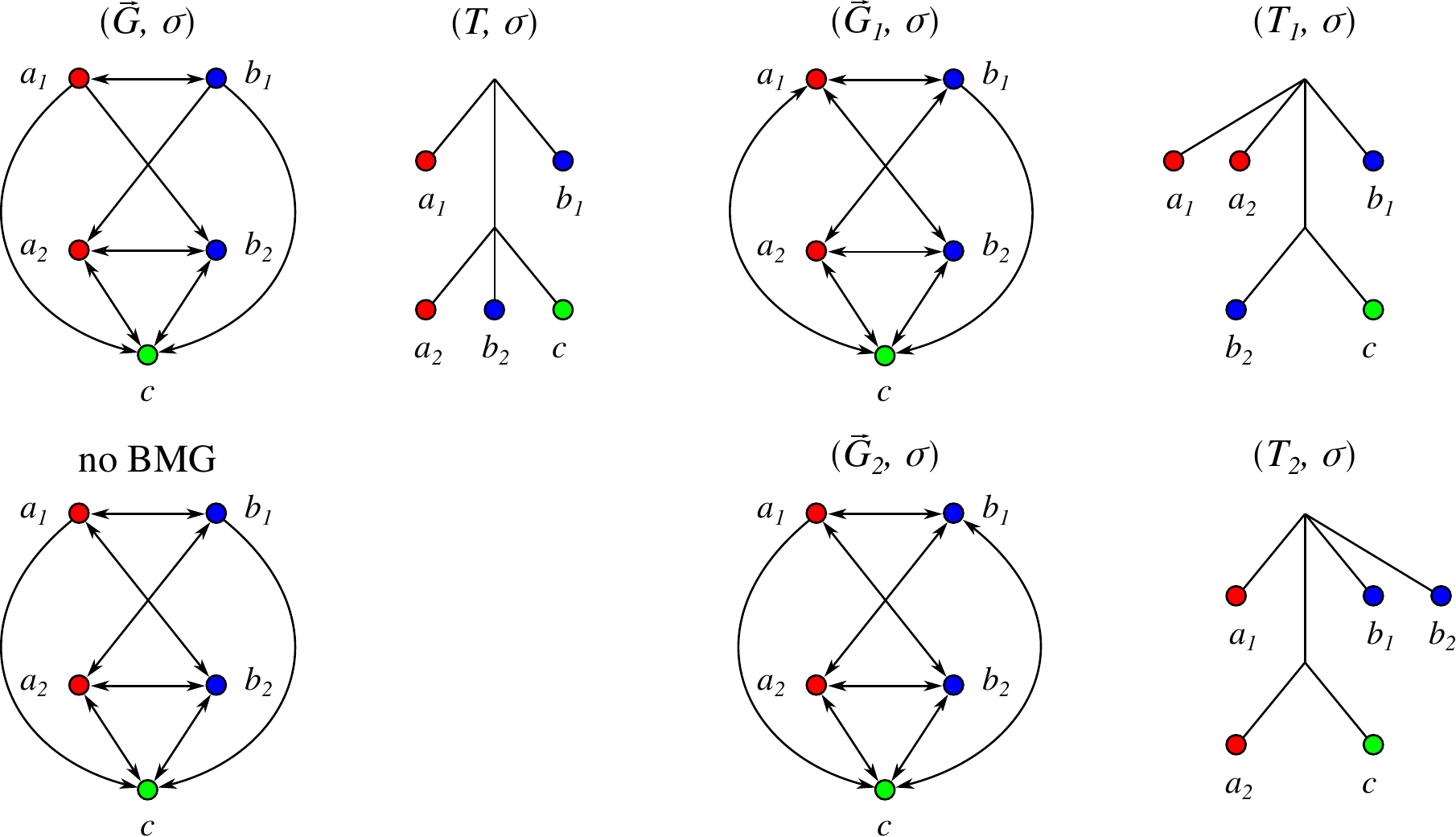}
  \caption{Example for \PROBLEM{3-BMG CBEG} with the 3-BMG $(\G,\sigma)$
    (explained by the LRT $(T,\sigma)$) as input that has no unique optimal
    solution. Insertion of the missing arcs $(a_2,b_1)$ and $(b_2,a_1)$
    produces a graph that is not a BMG. At least one of the arcs $(c,a_1)$
    or $(c,b_1)$ has to be inserted additionally to obtain the beBMGs
    $(\G_1,\sigma)$ and $(\G_2,\sigma)$ (shown with their LRTs
    $(T_1,\sigma)$ and $(T_2,\sigma)$), respectively.}
  \label{fig:CBEG-3-colors}
\end{figure}

In contrast to the 2-colored case, \PROBLEM{$\ell$-BMG CBEG} with a BMG as
input and $\ell\ge 3$ in general does not have a unique optimal solution.
In the example in Fig.~\ref{fig:CBEG-3-colors}, the missing arcs
$(a_2,b_1)$ and $(b_2,a_1)$ in the induced hourglass
$[a_1 b_1 \hourglass a_2 b_2]$ must be inserted. The resulting graph is not
a BMG. To obtain a BMG, it suffices to insert in addition either the 
arc $(c,a_1)$ or the arc $(c,b_1)$ to obtain a beBMG.
(cf. Prop.~\ref{prop:binary-iff-subtree-colors}).

The simple solution of \PROBLEM{$2$-BMG CBEG} begs the question whether
other arc modification problems for beBMGs, in particular the corresponding
deletion and editing problems, have a similar structure. This does not seem
to be case, however. Neither \PROBLEM{$2$-BMG EBEG} nor \PROBLEM{$2$-BMG
  DBEG} with a $2$-BMG as input have a unique optimal solution.  To see
this, consider the 2-BMG consisting of the hourglass $[xy\hourglass x'y']$
which is explained by the unique non-binary tree $(x,y,(x',y'))$ (in Newick
format, see also Fig.~\ref{fig:forb-sg-and-hourglass.pdf}(A)). Deletion of
the arcs $(x,y)$ or $(y,x)$ results in a graph that is explained by the
binary trees $(y,(x,(x',y')))$ or $(x,(y,(x',y')))$, respectively. We
suspect that a BMG as input does not make these problems easier than the
general case -- the complexity of which remains an open questions, however.

\paragraph{Author Contributions}
Conceptualization, methodology, formal analysis, and
writing, D.S., M.G., M.H., and P.F.S. All authors have read and agreed to
the published version of the manuscript.

\paragraph{Funding}
This research was funded in part by the German Research Foundation
(DFG), the Austrian Federal Ministries  BMK and BMDW and the Province of Upper 
Austria in the frame of the COMET Programme managed by FFG.

\paragraph{Conflicts of Interest}
The authors declare no conflict of interest.

\bibliography{preprint-BMG-to-beBMG}

\begin{thebibliography}{19}
\providecommand{\natexlab}[1]{#1}
\providecommand{\url}[1]{\texttt{#1}}
\expandafter\ifx\csname urlstyle\endcsname\relax
  \providecommand{\doi}[1]{doi: #1}\else
  \providecommand{\doi}{doi: \begingroup \urlstyle{rm}\Url}\fi

\bibitem[Abrams and Sklar(2010)]{Abrams:10}
G.~Abrams and J.~K. Sklar.
\newblock The graph menagerie: Abstract algebra and the mad veterinarian.
\newblock \emph{Math. Mag.}, 83:\penalty0 168--179, 2010.
\newblock \doi{10.4169/002557010X494814}.

\bibitem[Cohn et~al.(2002)Cohn, Pemantle, and Propp]{Cohn:02}
H.~Cohn, R.~Pemantle, and J.~G. Propp.
\newblock Generating a random sink-free orientation in quadratic time.
\newblock \emph{Electr. J. Comb.}, 9:\penalty0 R10, 2002.
\newblock \doi{10.37236/1627}.

\bibitem[Das et~al.(2020)Das, Ghosh, Ghosh, and Sen]{Das:20}
S.~Das, P.~Ghosh, S.~Ghosh, and S.~Sen.
\newblock Oriented bipartite graphs and the {Goldbach} graph.
\newblock Technical Report math.CO/1611.10259v6, arXiv, 2020.

\bibitem[DeSalle et~al.(1994)DeSalle, Absher, and Amato]{DeSalle:94}
R.~DeSalle, R.~Absher, and G.~Amato.
\newblock Speciation and phylogenetic resolution.
\newblock \emph{Trends Ecol. Evol.}, 9:\penalty0 297--298, 1994.
\newblock \doi{10.1016/0169-5347(94)90034-5}.

\bibitem[Gao et~al.(2013)Gao, Hare, and Nastos]{Gao:13}
Y.~Gao, D.~R. Hare, and J.~Nastos.
\newblock The cluster deletion problem for cographs.
\newblock \emph{Discrete Math.}, 313\penalty0 (23):\penalty0 2763--2771, 2013.
\newblock \doi{10.1016/j.disc.2013.08.017}.

\bibitem[Gei{\ss} et~al.(2019)Gei{\ss}, Ch{\'a}vez, Gonz{\'a}lez~Laffitte,
  L{\'o}pez~S{\'a}nchez, Stadler, Valdivia, Hellmuth, Hern{\'a}ndez~Rosales,
  and Stadler]{Geiss:19a}
M.~Gei{\ss}, E.~Ch{\'a}vez, M.~Gonz{\'a}lez~Laffitte, A.~L{\'o}pez~S{\'a}nchez,
  B.~M.~R. Stadler, D.~I. Valdivia, M.~Hellmuth, M.~Hern{\'a}ndez~Rosales, and
  P.~F. Stadler.
\newblock Best match graphs.
\newblock \emph{J. Math. Biol.}, 78:\penalty0 2015--2057, 2019.
\newblock \doi{10.1007/s00285-019-01332-9}.

\bibitem[Gei{\ss} et~al.(2020)Gei{\ss}, Stadler, and Hellmuth]{Geiss:19b}
M.~Gei{\ss}, P.~F. Stadler, and M.~Hellmuth.
\newblock Reciprocal best match graphs.
\newblock \emph{J. Math. Biol.}, 80:\penalty0 865--953, 2020.
\newblock \doi{10.1007/s00285-019-01444-2}.

\bibitem[Hellmuth et~al.(2015)Hellmuth, Wieseke, Lechner, Lenhof, Middendorf,
  and Stadler]{Hellmuth:15a}
M.~Hellmuth, N.~Wieseke, M.~Lechner, H.-P. Lenhof, M.~Middendorf, and P.~F.
  Stadler.
\newblock Phylogenetics from paralogs.
\newblock \emph{Proc. Natl. Acad. Sci. USA}, 112:\penalty0 2058--2063, 2015.
\newblock \doi{10.1073/pnas.1412770112}.

\bibitem[Hoelzer and Meinick(1994)]{Hoelzer:94}
G.~A. Hoelzer and D.~J. Meinick.
\newblock Patterns of speciation and limits to phylogenetic resolution.
\newblock \emph{Trends Ecol Evol}, 9:\penalty0 104--107, 1994.
\newblock \doi{10.1016/0169-5347(94)90207-0}.

\bibitem[Korchmaros(2020{\natexlab{a}})]{Korchmaros:20a}
A.~Korchmaros.
\newblock The structure of 2-colored best match graphs.
\newblock Technical Report math.CO/2009.00447v2, arXiv, 2020{\natexlab{a}}.

\bibitem[Korchmaros(2020{\natexlab{b}})]{Korchmaros:20b}
A.~Korchmaros.
\newblock Circles and paths in 2-colored best match graphs.
\newblock Technical Report math.CO/2006.04100v1, arXiv, 2020{\natexlab{b}}.

\bibitem[Liu et~al.(2011)Liu, Wang, Guo, and Chen]{Liu:11}
Y.~Liu, J.~Wang, J.~Guo, and J.~Chen.
\newblock Cograph editing: Complexity and parametrized algorithms.
\newblock In B.~Fu and D.~Z. Du, editors, \emph{COCOON 2011}, volume 6842 of
  \emph{Lect. Notes Comp. Sci.}, pages 110--121, Berlin, Heidelberg, 2011.
  Springer-Verlag.

\bibitem[Maddison(1989)]{Maddison:89}
W.~Maddison.
\newblock Reconstructing character evolution on polytomous cladograms.
\newblock \emph{Cladistics}, 5:\penalty0 365--377, 1989.
\newblock \doi{10.1111/j.1096-0031.1989.tb00569.x}.

\bibitem[Natanzon et~al.(2001)Natanzon, Shamir, and Sharan]{Natanzon:2001}
A.~Natanzon, R.~Shamir, and R.~Sharan.
\newblock Complexity classification of some edge modification problems.
\newblock \emph{Discr. Appl. Math.}, 113:\penalty0 109--128, 2001.
\newblock \doi{10.1016/S0166-218X(00)00391-7}.

\bibitem[Schaller et~al.(2021{\natexlab{a}})Schaller, Gei{\ss}, Hellmuth, and
  Stadler]{Schaller:21c}
D.~Schaller, M.~Gei{\ss}, M.~Hellmuth, and P.~F. Stadler.
\newblock Best match graphs with binary trees.
\newblock In C.~Mart{\'\i}n-Vide, M.~A. Vega-Rodr{\'\i}guez, and T.~Wheeler,
  editors, \emph{Algorithms for Computational Biology, 8th AlCoB}, Lect. Notes
  Comp. Sci., 2021{\natexlab{a}}.
\newblock in press; arXiv 2011.00511.

\bibitem[Schaller et~al.(2021{\natexlab{b}})Schaller, Gei{\ss}, Stadler, and
  Hellmuth]{Schaller:21a}
D.~Schaller, M.~Gei{\ss}, P.~F. Stadler, and M.~Hellmuth.
\newblock Complete characterization of incorrect orthology assignments in best
  match graphs.
\newblock \emph{J. Math. Biol.}, 82:\penalty0 20, 2021{\natexlab{b}}.
\newblock \doi{10.1007/s00285-021-01564-8}.

\bibitem[Schaller et~al.(2021{\natexlab{c}})Schaller, Geiß, Hellmuth, and
  Stadler]{Schaller:21x}
D.~Schaller, M.~Geiß, M.~Hellmuth, and P.~F. Stadler.
\newblock Least resolved trees for two-colored best match graphs.
\newblock 2021{\natexlab{c}}.
\newblock submitted; arxiv: 2101.07000.

\bibitem[Schaller et~al.(2021{\natexlab{d}})Schaller, Stadler, and
  Hellmuth]{Schaller:21b}
D.~Schaller, P.~F. Stadler, and M.~Hellmuth.
\newblock Complexity of modification problems for best match graphs.
\newblock \emph{Theor. Comp. Sci.}, 2021{\natexlab{d}}.
\newblock \doi{10.1016/j.tcs.2021.02.037}.

\bibitem[Slowinski(2001)]{Slowinski:01}
J.~B. Slowinski.
\newblock Molecular polytomies.
\newblock \emph{Mol. Phylog. Evol.}, 19:\penalty0 114--120, 2001.
\newblock \doi{10.1006/mpev.2000.0897}.

\end{thebibliography}


\end{document}